\documentclass[10pt,twocolumn]{article}

\usepackage[left=20mm,right=20mm,top=20mm,bottom=25mm,a4paper]{geometry}
\usepackage[utf8]{inputenc}

\usepackage{amssymb}
\usepackage{amsmath}
\usepackage{amsthm}
\usepackage{nicefrac}
\usepackage{bbm}
\theoremstyle{definition}   
\usepackage{chngcntr}
\usepackage{apptools}
\usepackage{relsize}
\usepackage{multirow}
\usepackage{mathtools}

\usepackage[usenames,dvipsnames,svgnames]{xcolor}
\usepackage{tikz}
\usepackage{pgfplots}
\usepgfplotslibrary{groupplots}
\usepackage{pgfplotstable}
\usepackage{subcaption}
\usepackage[font=small]{caption}
\usepackage{wrapfig}
\pgfplotsset{compat=newest}
\usepgfplotslibrary{fillbetween}
\usetikzlibrary{patterns}
\usepgfplotslibrary{colormaps}
\usetikzlibrary{backgrounds}

\usetikzlibrary{shapes,positioning,arrows.meta} 

\tikzset{>=latex} 

\usepackage{enumitem}

\usepackage[maxbibnames=50, maxcitenames=2, uniquelist=false,
    natbib=true, style=authoryear-comp,
    isbn=false, url=false,giveninits=true,dashed=false]{biblatex}   
\bibliography{bibliography}
\AtBeginBibliography{} 

\usepackage[main=english]{babel}
\usepackage[autostyle]{csquotes}
\usepackage[T1]{fontenc}
\usepackage[hidelinks]{hyperref}
\usepackage[affil-it]{authblk}
\usepackage[capitalise]{cleveref}
\crefname{section}{Sect.}{Sect.}
\crefname{appendix}{App.}{App.}

\usepackage{booktabs}

\newcommand{\psiN}{\psi_0}
\newcommand{\psiNnn}{\psi_0^{\text{NN}}}
\newcommand{\psiNnnt}{\widetilde{\psi}{^0_{\text{NN}}}}
\newcommand{\psiM}{\psi^{\text{max}}_0}
\newcommand{\Dpsi}{\Delta\psi}

\newcommand{\IoneISO}{\Bar{K}_1}
\newcommand{\ItwoISO}{\Bar{K}_2}

\newcommand{\SP}{\mathcal{SP}}
\newcommand{\SM}{\mathcal{SM}}
\newcommand{\TANH}{\mathcal{TH}}

\newcommand{\cof}{\operatorname{cof}}

\newcommand{\Cross}{\mathbin{\tikz [x=1.4ex,y=1.4ex,line width=.25ex] \draw (0.1,0.1) -- (0.9,0.9) (0.1,0.9) -- (0.9,0.1);}}
\newcommand{\norm}[1]{\left\lVert#1\right\rVert}

\newcommand{\ba}{\boldsymbol{a}}
\newcommand{\bb}{\boldsymbol{b}}
\newcommand{\bs}{{s}}
\newcommand{\bn}{\boldsymbol{n}}
\newcommand{\bu}{\boldsymbol{u}}
\newcommand{\bw}{\boldsymbol{w}}
\newcommand{\bx}{\boldsymbol{x}}

\newcommand{\bF}{\boldsymbol{F}}
\newcommand{\bH}{\boldsymbol{H}}
\newcommand{\bI}{\boldsymbol{I}}
\newcommand{\bP}{\boldsymbol{P}}
\newcommand{\bW}{\boldsymbol{W}}

\newcommand{\bbA}{{\mathbb{A}}}
\newcommand{\bbR}{{\mathbb{R}}}

\newcommand{\deriv}{\mathrm{d}}
\newcommand{\derivF}{\deriv_{\bF}}
\newcommand{\derivFF}{\deriv^2_{\bF\bF}}

\newcommand{\bbF}{\Bar{\bF}}

\newcommand{\ab}{(\ba\otimes\bb)}

\newcommand{\bPeq}{\bP^{\text{eq}}}
\newcommand{\bFeq}{\bF^{\text{eq}}}
\newcommand{\bHeq}{\bH^{\text{eq}}}
\newcommand{\bxeq}{\bx^{\text{eq}}}
\newcommand{\bueq}{\bu^{\text{eq}}}
\newcommand{\peq}{p^{\text{eq}}}
\newcommand{\bxn}{\bx_0}
\newcommand{\bbAeq}{\bbA^{\text{eq}}}

\newcommand{\bnull}{\boldsymbol{0}}

\newcommand{\cD}{{\mathcal{D}}}

\newcommand{\cF}{{\mathcal{F}}}

\tikzset{
    ncbar angle/.initial=90,
    ncbar/.style={
        to path=(\tikztostart)
        -- ($(\tikztostart)!#1!\pgfkeysvalueof{/tikz/ncbar angle}:(\tikztotarget)$)
        -- ($(\tikztotarget)!($(\tikztostart)!#1!\pgfkeysvalueof{/tikz/ncbar angle}:(\tikztotarget)$)!\pgfkeysvalueof{/tikz/ncbar angle}:(\tikztostart)$)
        -- (\tikztotarget)
    },
    ncbar/.default=0.5cm,
}

\tikzset{square left brace/.style={ncbar=0.25cm}}
\tikzset{square right brace/.style={ncbar=-0.25cm}}

\definecolor{CPSgreen}{RGB}{22,164,138}
\definecolor{CPSlightblue}{RGB}{104,143,198}
\definecolor{CPSdarkblue}{RGB}{67,83,132}
\definecolor{CPSgrey}{RGB}{204, 204, 204}
\definecolor{CPSdarkgrey}{RGB}{175, 175, 175}
\definecolor{CPSorange}{RGB}{246,163,21}
\definecolor{CPSred}{RGB}{194,76,76}

\newtheorem{definition}{Definition}[section]
\newtheorem{theorem}[definition]{Theorem}
\newtheorem{corollary}[definition]{Corollary}
\newtheorem{remark}[definition]{Remark}

\title{Constrained neural softening functions for the Mullins effect in parameter-dependent materials}

\author[*]{Dominik~K.~Klein}
\author[ ]{Jonas~Fey}
\author[ ]{Maximilian~Kannapinn}
\author[ ]{Oliver~Weeger}

\affil[ ]{\footnotesize Cyber-Physical Simulation, 
Department of Mechanical Engineering, TU Darmstadt, 64293 Darmstadt, Germany}

\affil[*]{\footnotesize Corresponding author, email: klein@cps.tu-darmstadt.de}

\date{\today}

\begin{document}

\twocolumn[
  \begin{@twocolumnfalse}
    \maketitle
    \setlength{\leftmargini}{1em}
    \vspace{-4ex}
\par\noindent\rule{\linewidth}{0.4pt}

\begin{abstract}
\normalsize
\noindent
Accurately capturing the Mullins effect across varying material compositions remains challenging, as constitutive models must unify flexible representations of stress softening with physical consistency and robust parameter calibration. In this work, we investigate parameter-dependent physics-augmented neural network constitutive models within the pseudo-elastic modeling framework, focusing on structural constraints that ensure physically admissible neural softening functions. \\
This paper has three \textbf{key contributions}: \textbf{(1)} We formulate different neural softening functions within the pseudo-elastic framework, including a multiplicative decomposition of the softening function that enables parameter-dependent ansatz functions satisfying common mechanical constitutive conditions by construction. \textbf{(2)} We derive a sufficient condition for ellipticity of pseudo-elastic constitutive models, which can be enforced in the constitutive models by construction. \textbf{(3)} We systematically investigate how flexibility and structural constraints influence the performance of neural softening functions in terms of prediction accuracy and calibration robustness. For this, we consider multiaxial experimental data of carbon black-filled EPDM rubber synthesized with varying amounts of curatives. The findings of our work suggest that neural softening functions are a promising extension of the classical pseudo-elastic framework and, when subject to suitable constraints, can improve predictive accuracy for highly nonlinear, parameter-dependent materials.

\end{abstract}
\vspace*{2ex}

\par\noindent\rule{\linewidth}{0.4pt}
  \end{@twocolumnfalse}
]


\section{Introduction}

\textbf{Neural constitutive models} have emerged in recent years as a promising approach to representing complex material behaviors. There is broad agreement that such models should incorporate mechanical constitutive conditions to provide mechanically reasonable predictions, while exploiting the flexibility of neural networks.
This is referred to as thermodynamics-based \parencite{masi2021}, mechanics-informed \parencite{ASAD2023116463}, physics-based \parencite{aldakheel2025}, physics-constrained \parencite{kalina2022b}, physics-augmented \parencite{klein2022b}, or constitutive artificial \parencite{Linka2020} neural network (NN) modeling. Throughout this work, we adopt the terminology \emph{physics-augmented neural networks} (PANNs). Beyond, the paradigm of combining machine learning methods with physical and mathematical knowledge is widely used in many scientific fields \parencite{rueden2021,peng2021,karniadakis2021,Kumar2022,kannapinn,watson2025}. 

As in classical approaches \parencite{TruesdellNoll,neff2015}, an important question is which constraints should be imposed on the constitutive models and how these affect the resulting model behavior. Considerable progress has been made in hyperelastic PANNs \parencite{fuhg2024}, and various approaches have been proposed for inelasticity \parencite{Rosenkranz2024,boes2025accountingplasticityextensioninelastic,FRIEDRICHS2026118852} and multiphysics \parencite{klein2022b,kalina2024a,HOLTHUSEN2026119082}. Moreover, neural constitutive models that incorporate additional parameters that represent families of material behaviors, e.g., the geometry of microstructured materials, material constituents, or manufacturing conditions, have also been established \parencite{Linka2020,fernandezMaterialModelingParametric2022,ABDOLAZIZI2025106212,Klein2026a}. In the present work, we focus on pseudo-elastic constitutive modeling of the Mullins effect and investigate suitable constraints on neural softening functions.

\medskip

The \textbf{Mullins effect} describes stress softening under mechanical loading \parencite{mullins1948}. Thereby, softening mainly develops upon loading to previously unattained deformation levels, although additional softening may occur during subsequent loading cycles even without exceeding the previously attained maximum deformation \parencite{Plagge_Klüppel_2017}. The Mullins effect exhibits further complexities such as temperature dependence \parencite{Plagge_Klüppel_2017}, recovery over time \parencite{LIAO2023104338}, softening-induced anisotropy \parencite{DENORA2024108399}, and it is frequently accompanied by other inelastic phenomena such as viscoelasticity \parencite{ALKHOURY2024105650}. The Mullins effect can be observed in a variety of soft materials, including rubber-like materials \parencite{harwood1966,mullins1948} and soft biological tissue \parencite{FRANCESCHINI20062592,KUMAR2025104406}. Moreover, it is also observed across families of materials with varying compositions or manufacturing conditions \parencite{LIAO2020103366,Plagge_Klüppel_2017}. 

Although being most pronounced during initial loading, the consideration of the Mullins effect is important for reliable constitutive model predictions. Mechanical preconditioning, i.e., trying to remove the influence of the Mullins effect before experimentally characterizing a material specimen, can significantly affect the predictive capabilities of constitutive models \parencite{WOLLNER2026106698}. Moreover, heterogeneous deformation states in engineering components can lead to spatially heterogeneous softening, even for initially homogeneous materials \parencite[Fig.~7.6]{freund2012}. 

\medskip

A variety of \textbf{classical constitutive models for the Mullins effect} have been proposed, including both micromechanically motivated models \parencite{Plagge_Klüppel_2017,GOKTEPE20052259} and phenomenological approaches, such as continuum damage \parencite{SIMO1987153,miehe1995} or pseudo-elastic formulations \parencite{DORFMANN20041855,ogden_roxburgh}, cf.~\parencite{ricker2021}. Depending on the considered phenomena, these models describe isotropic softening \parencite{ogden_roxburgh} or anisotropic softening \parencite{itskov2010,DIANI2006483}, and may additionally incorporate further effects such as recovery \parencite{WANG201888}, temperature dependence \parencite{PLAGGE2020103291}, or viscoelasticity \parencite{CARO2026105818}. Moreover, models that include parametric dependencies of the material behavior have also been introduced \parencite{LIAO2020103366}. 

\medskip

\textbf{Neural constitutive modeling for the Mullins effect} has received only limited attention. To the best of our knowledge, the first neural constitutive model for the Mullins effect was proposed by \textcite{Ghaderi_Morovati_Dargazany_2020}. However, the thermodynamic consistency argument employed therein, which is based on monotonicity of the NNs, is not satisfied by the employed non-monotonic sine activation functions, cf.~App.~B.2 and~Sec.~4 therein. Neural constitutive models based on continuum damage formulations were proposed by~\textcite{zlatic2024} and~\textcite{TAC2024102220}, where the former employ feed-forward neural networks and the latter employ neural ordinary differential equations. In both cases, the stress response is scaled by a damage function such that softening also occurs along primary loading paths. Moreover, in~\textcite{OSTADRAHIMI2026104582}, a neural constitutive model that accounts for viscoelasticity, temperature dependency, and softening behavior is introduced, however, thermodynamic consistency is only approximately fulfilled by promoting it through the loss function. Closely related, pseudo-elastic constitutive models based on deep symbolic regression were introduced in~\textcite{ABDUSALAMOV2024104037}. Overall, neural constitutive models for the Mullins effect are still scarce, thermodynamic consistency is not always ensured by construction, and which constraints should be imposed on the neural softening or damage functions, as well as how these affect the resulting model behavior, has not yet been systematically investigated.

\medskip

In the \textbf{present work}, we focus on the classical pseudo-elastic modeling framework proposed by~\textcite{ogden_roxburgh},\footnote{For discussions on the thermodynamics of pseudo-elasticity, the reader is referred to \textcite{DESOUZANETO19941533} and \textcite{Naumann2015}.} which consists of two components: an energy potential describing the unsoftened material behavior, whose gradient determines the unsoftened stress response, and a softening function that scales this stress response to account for the Mullins effect. We consider isotropic materials that remain isotropic under softening and neglect additional effects such as recovery. Furthermore, we consider parameter-dependent material behavior, i.e., families of materials characterized by varying material constituents, microstructures, or manufacturing parameters. 

The first component of the pseudo-elastic framework, i.e., the energy potential, is well-established in neural constitutive modeling \parencite{linden2023,Geuken2026,Linka2023}, including parametrized formulations \parencite{Klein2026a,Linka2020}. In contrast, to the best of our knowledge, neural softening functions within the pseudo-elastic framework have not yet been proposed. In this work, we formulate and evaluate different neural softening functions within the pseudo-elastic framework. In particular, we introduce a multiplicative decomposition of the softening function that enables parameter-dependent neural ansatz functions satisfying common mechanical constitutive conditions by construction. We systematically investigate how flexibility and structural constraints influence the performance of neural softening functions in terms of prediction accuracy and calibration robustness. To this end, we consider multiaxial experimental data of carbon black-filled EPDM rubber synthesized with varying amounts of curatives by \textcite{Plagge_Klüppel_2017}.

\medskip

Moreover, to the best of our knowledge, \textbf{ellipticity of pseudo-elastic constitutive models} has not yet explicitly been investigated. Ellipticity is associated with the propagation of waves in a solid body and frequently linked to a stable behavior of numerical simulations \parencite{Schroeder_Neff_Balzani_2005,Zee1983}. Closely related, \textcite[Rem.~1]{Naumann2015} noted that, under certain conditions on the softening function and in some loading scenarios, a pseudo-elastic model inherits the polyconvexity of the underlying energy potential. However, the existence theorems associated with polyconvexity \parencite{Ball1976,Ball1977} are formulated for hyperelasticity and do not immediately carry over to pseudo-elastic models. Still, to anticipate~App.~\ref{sec:ellip}, this observation by Naumann~and~Ihlemann already entails ellipticity of pseudo-elastic constitutive models for \emph{some} loading regimes. Stability of pseudo-elastic models under uniaxial tension has been discussed by \textcite[Eq.~(3.18)]{ogden_roxburgh}. In the present work, we derive a sufficient condition for ellipticity of pseudo-elastic constitutive models, accounting for the different tangent operators associated with primary loading and unloading or reloading. We translate this condition into a monotonicity constraint on the neural softening function, which can be imposed by construction.

\medskip

The outline of the remaining manuscript is as follows. In~\cref{sec:basics}, we introduce the fundamentals of parameterized pseudo-elasticity. We then discuss the formulation of neural softening functions in~\cref{sec:PANN}, which are applied to the experimental data in~\cref{sec:application}. After the conclusion in~\cref{sec:conc}, we present the derivation of the ellipticity condition for pseudo-elastic Mullins modeling in~App.~\ref{sec:ellip}.

\section{Fundamentals of parameterized pseudo-elasticity}\label{sec:basics}

\begin{figure*}[t]
\centering
    \includegraphics[width=0.8\textwidth]{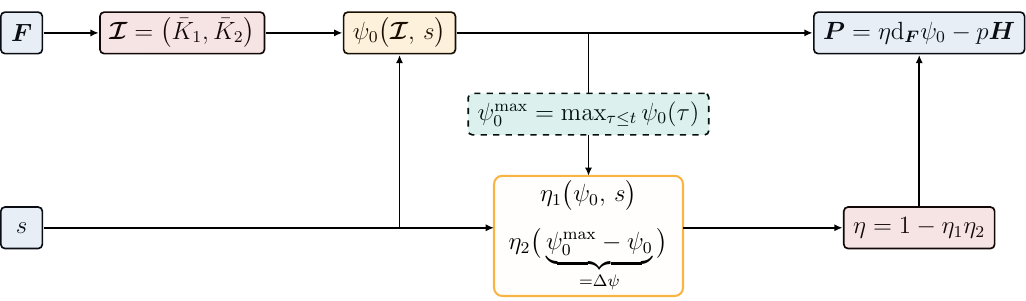}
\caption{Overview of the constitutive modeling framework. In this work, the energy potential $\psiN$ and the softening function composed of $\eta_1,\,\eta_2$ are represented by neural networks.}
\label{fig:model_overview}
\end{figure*}

We now introduce the fundamentals of pseudo-elastic Mullins softening, following the framework of \textcite{ogden_roxburgh}. For discussions on the thermodynamics of pseudo-elasticity, the reader is referred to \textcite{DESOUZANETO19941533} and \textcite{Naumann2015}. 

Regarding the {kinematics}, $\bF\in\text{GL}^+(3)$ denotes the deformation gradient, $\bH=\cof\bF\in\text{GL}^+(3)$ its cofactor, and $J=\det\bF\in\bbR^+$ its determinant, where $\text{GL}^+(3)$ is the general linear group in three dimensions. Applying the decomposition proposed by \textcite{flory}, the isochoric deformation gradient and its cofactor are given by $\Bar{\bF}=J^{-1/3}\bF\in\text{SL}(3)$ and $\Bar{\bH}=\cof\Bar{\bF}\in\text{SL}(3)$, where $\text{SL}(3)$ denotes the special linear group in three dimensions. We further introduce the isochoric isotropic invariants $\IoneISO=\norm{\Bar{\bF}}$ and $\ItwoISO=\norm{\Bar{\bH}}$, corresponding to the square roots of the typically isochoric invariants, $\bar{I}_1=\bar{K}_1^2$ and $\bar{I}_2=\bar{K}_2^2$, see also Remark \ref{remark:invar}.


In pseudo-elasticity, the first Piola-Kirchhoff stress is given by\footnote{For a derivation of this stress formula from a thermodynamic potential, the reader is referred to \textcite[Sec.~3.3]{Naumann2015} or \textcite[Prop.~2.1]{DESOUZANETO19941533}.}
\begin{equation}\label{eq:total_stress}
\bP=\eta \,\derivF\psiN -p \bH\,,
\end{equation}
where $p\in\bbR$ is a Lagrange multiplier ensuring incompressibility ($J=1$). Moreover, $\psiN$ is the virgin, unsoftened energy potential and $\eta$ is the softening function. We assume an incompressible body which is isotropic and remains so under softening. We furthermore assume a \textbf{parameterized} material behavior characterized by the parameter $\bs\in\bbR$. In the following, we introduce parameterized formulations for $\psiN$ and $\eta$ and discuss how they can be restricted so that the constitutive modeling framework satisfies relevant constitutive conditions.

\begin{remark} \label{remark:invar}
The use of $\IoneISO,\,\ItwoISO$ instead of the commonly employed invariants $\bar{I}_1,\,\bar{I}_2$ can provide a polyconvex constitutive model with more flexibility, cf.~\textcite[Sec.~3.2]{WOLLNER2027106814}. 
\end{remark}
\begin{remark} \label{remark:param}
For the examples in this work, we consider a scalar parameter $\bs\in\bbR$. However, both the theoretical framework and the PANN constitutive model could be extended to vector-valued parametrizations $\boldsymbol{s}\in\bbR^m$, cf.~\textcite{Klein_Roth_Valizadeh_Weeger_2023}. 
\end{remark}

\subsection{Energy potential}\label{sec:ener_pot}

We consider energy potentials of the form
\begin{equation}\label{eq:pot_invs}
\begin{aligned}
\psiN\left(\IoneISO(\bbF),\,\ItwoISO(\bbF),\,\bs\right)\,,
\end{aligned}
\end{equation}
with $\psiN\colon\bbR^+\times\bbR^+\times \bbR\rightarrow\bbR^+$, which depend on the deformation through the invariants $\IoneISO$ and $\ItwoISO$ introduced above, as well as on the parameter $\bs$. 
We assume that the energy potential is a \textbf{polyconvex} function in the sense of \textcite{Ball1976,Ball1977}. To ensure polyconvexity, we employ the sufficient (but not necessary) conditions that $\psiN$ is convex and component-wise monotonically increasing in $(\IoneISO,\,\ItwoISO)$, which, in the incompressible case, are polyconvex invariants, cf.~\textcite[Sec.~3.2]{WOLLNER2027106814} and \textcite[Sec.~2.3]{Klein2026a}. Then, the only \textbf{normalization condition} we need to impose explicitly on the energy potential is that it vanishes in the undeformed configuration, i.e., $\psiN\rvert_{\bF=\bI}=0$. As we assume $\psiN$ to be monotonically increasing in the invariants, and the invariants attain their global minimum in the undeformed configuration $\bF=\bI$ \parencite[Corr.~A.11]{Schroeder2003}, this also ensures the non-negativity of the energy, i.e., $\psiN\geq0$. Finally, since $\psiN$ depends on the isochoric isotropic invariants $(\IoneISO,\,\ItwoISO)$, stress normalization is fulfilled a priori, i.e., $\derivF\psiN\rvert_{\bF=\bI}=\bnull$ \parencite{sansour2018}.

\subsection{Softening function}\label{sec:soft}

The softening function depends on the deformation through $\psiN$ and on the deformation history through
\begin{equation}\label{eq:psiM}
  \psiM(t)=\max_{\tau\leq t}\psiN(\tau)\,,
\end{equation}
where, with a slight abuse of notation, $\psiN(\tau)$ denotes the value of the energy potential at time $\tau$. This means that $\psiM$ is the maximum value of the energy potential experienced up to the current time $t$. We further define
\begin{equation}\label{eq:Dpsi}
  \Dpsi=\psiM-\psiN\in\mathbb{R}^+\,,
\end{equation}
which measures the difference of the current value of the energy potential from its maximum value in the deformation history. Based on this, we consider softening functions of the form
\begin{equation}
  \eta\big(\psiN,\psiM,\bs\big)=1-\eta_1(\psiN,s)\,\eta_2(\Dpsi)\,,
\end{equation}
with $\eta_1\colon\bbR^+\times\bbR\rightarrow[0,1]$ and $\eta_2\colon\bbR^+\rightarrow[0,1]$, which implies $\eta\in[0,1]$. As will be demonstrated in the remainder of this work, this multiplicative decomposition into $\eta_1$ and $\eta_2$ facilitates the formulation of suitable softening functions.

As $\eta$ depends on the deformation only through $\psiN$ and $\psiM$, \textbf{thermodynamic consistency} can be fulfilled by the monotonicity constraint
\begin{equation}\label{eq:eta_thermo_cons}
  \partial_{\Dpsi}\eta_2 \geq 0 \quad\Longrightarrow\quad  \partial_{\psiM}\eta \leq 0\,,
\end{equation}
together with the assumed property $\eta_1\geq 0$, cf.~\textcite[Eq.~(37)]{Naumann2015}. The \textbf{normalization condition}
\begin{equation}\label{eq:eta_norm}
    \eta_2\big\rvert_{\Dpsi=0}=0   \quad\Longrightarrow  \quad\eta\big\rvert_{\psiN=\psiM}=1\,,
\end{equation}
ensures that the stress response is not softened whenever the energy potential attains its maximum value in the deformation history. As a consequence, $\psiN$ represents the \emph{unsoftened} material behavior \parencite[Sect.~3]{Naumann2015}.

Recall that the energy potentials employed in this work are polyconvex functions. Then, as we derive in~App.~\ref{sec:ellip}, a sufficient (but not necessary) condition for \textbf{ellipticity} of the overall constitutive model is the monotonicity constraint
\begin{equation}\label{eq:eta_ellip}
\partial_{\psiN}\eta_{1}\leq 0\quad\Longrightarrow\quad 
       \partial_{\psiN} \eta\geq 0\,,
\end{equation}
together with the assumed property $\eta_1,\eta_2\geq 0$. 

Finally, the energy potential $\psiN$ is an objective and isotropic function, and the softening function $\eta$ depends on the deformation only through $\psiN$ and on the deformation history only through $\psiM$. Thus, the stress $\bP$ in~\cref{eq:total_stress} satisfies \textbf{objectivity} and \textbf{material symmetry} in the form of isotropy a priori \parencite{Holzapfel2000}. The overall structure of the modeling framework is visualized in~\cref{fig:model_overview}. 

\begin{remark}
The sufficient condition for ellipticity introduced in~\cref{eq:eta_ellip}, i.e., $\partial_{\psiN}\eta_1\leq 0$, provides a physical motivation for imposing a corresponding monotonicity constraint on $\eta_1$. However, in contrast to thermodynamic consistency, ellipticity is not imposed as a mandatory constitutive requirement in the present work. Indeed, constitutive models that do not satisfy ellipticity by construction are commonly employed \parencite{Klein2026a}, although loss of ellipticity may have important implications for material stability and the well-posedness of boundary-value problems. We therefore also consider models for which the sufficient condition in~\cref{eq:eta_ellip} is not enforced by construction.
\end{remark}

\subsection{Parametric dependency}\label{sec:param}

Besides requiring the constitutive conditions discussed above to hold for every value of the parameter $\bs$, no physical considerations restrict how the energy potential or the softening function depend on $\bs$. To provide the energy potential with some structure in terms of its parametric dependency, we employ the heuristic monotonicity condition $\partial_{\bs}\psiN\geq0$, which has proven useful when applying hyperelastic constitutive models to experimental data of soft rubber-like materials \parencite{ABDOLAZIZI2025106212,Klein2026a}.

Likewise, the dependence of the softening function on the parameter $\bs$ can be restricted by imposing monotonicity constraints. Recall that $\bs$ enters the softening function through $\eta_1(\psiN,\bs)$. In this work, we consider three variants for $\eta_1$: (i) an unrestricted function, (ii) a monotonically increasing function with $\partial_{\psiN}\eta_1\geq0$ and $\partial_{\bs}\eta_1\geq0$, and (iii) a monotonically decreasing function with $\partial_{\psiN}\eta_1\leq0$ and $\partial_{\bs}\eta_1\leq0$. For variants (i) and (ii), the sufficient condition for ellipticity introduced in~\cref{eq:eta_ellip} is not satisfied by construction. Mixed monotonicity constraints, such as $\partial_{\psiN}\eta_1\geq0$ and $\partial_{\bs}\eta_1\leq0$, could also be imposed. However, the three considered variants already allow us to investigate the influence of monotonicity constraints on the predictive performance of the resulting models. As the present investigation is limited to a single experimental dataset, we refrain from introducing additional variants and leave a more comprehensive investigation of mixed monotonicity constraints to future work.

\begin{remark}
Throughout this work, to represent the energy potential, we adopt the physics-augmented neural network model proposed by \textcite[Sect.~4]{Klein_Roth_Valizadeh_Weeger_2023}. This allows us to focus the remainder of this work on the softening function. In particular, we employ
\begin{equation}\label{eq:NN_energy}
    \begin{aligned}     \bx^{(1)}&=\TANH\big(\bw^{(1)}\,\bs+\bb^{(1)}\big)\quad &&\in\bbR^n
        \\
        \bx^{(2)}&=\SP\big(\bW^{(2)}\cdot(\IoneISO-\sqrt{3},
        \\
       &\phantom{=\SP\big(} \ItwoISO-\sqrt{3},\,\bx^{(1)})^T+\bb^{(2)}\big)\quad &&\in\bbR^n
        \\
        \psiNnnt&=\bw^{(3)}\cdot\bx^{(2)}\quad &&\in\bbR\,,
        \\
            \psiNnn&=\psiNnnt-\psiNnnt\big\rvert_{\bF=\bI}\quad &&\in\bbR\,,
    \end{aligned}
\end{equation}
where $\TANH$ denotes the hyperbolic tangent function and $\SP$ denotes the softplus function. With the restriction $w^{(1)}_i,\,W^{(2)}_{ij},\,w^{(3)}_i\geq 0$, this energy potential is convex and monotonically increasing in polyconvex invariants and thus polyconvex, and monotonically increasing in $\bs$. Further, it satisfies all normalization constraints on $\psiN$ introduced above. The construction of such feed-forward neural networks incorporating convexity and monotonicity constraints was proposed by \textcite{Amos2017}, and its use to construct polyconvex neural energy potentials was introduced by \textcite{klein2022a}.
\end{remark}

\section{Formulation of neural softening functions}\label{sec:PANN}

We now discuss different formulations of neural softening functions. We start by introducing possible neural network (NN) architectures in~\cref{sec:NN_ansatz}. The corresponding neural softening functions are introduced in~\cref{sec:neural_softening_models}.

\begin{figure}[t!]
    \centering        
    \resizebox{0.69\linewidth}{!}{
    \includegraphics{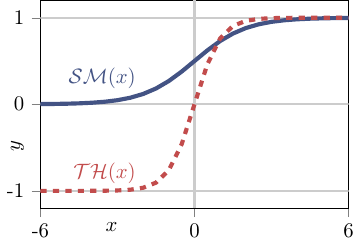}}
 \caption{The hyperbolic tangent function $\TANH$ is monotonically increasing, and, for $x \geq 0$, takes values between 0 and 1. The sigmoid function $\SM$ is monotonically increasing and takes values between 0 and 1.}
        \label{fig:act_fcts}
\end{figure}

\subsection{Neural network architectures}\label{sec:NN_ansatz}

We consider scalar-valued feed-forward neural networks (FFNNs) given as the mapping $\cF:\bbR^{n_0}\rightarrow \bbR\,,\,\, \bx^{(0)}\mapsto \cF(\bx^{(0)})$, see e.g.~\parencite{herrmann2025} for a general introduction to NNs. In particular, we consider single-layered FFNNs given as\footnote{Note that generally, FFNNs with an arbitrary amount of layers could be employed to represent neural softening functions.}
\begin{equation}\label{eq:HL1}
    \cF=\sigma\big(\bw^{(1)}\cdot \bx^{(0)}+b^{(1)}\big) \in\bbR\,,
\end{equation}
and two-layered FFNNs given as
\begin{equation}\label{eq:HL2}
        \begin{aligned}
\bx^{(1)}&=\sigma\big(\bW^{(1)}\cdot \bx^{(0)}+\bb^{(1)}\big) &&\in\bbR^{n_1}
\\
\cF&=\sigma\big(\bw^{(2)}\cdot\bx^{(1)}+b^{(2)}\big) &&\in\bbR\,.
\end{aligned}
\end{equation}
Here, $\sigma\in\{\SM,\,\TANH\}$ denotes the component-wise applied activation function, for which we employ either the hyperbolic tangent function $\TANH$ or the sigmoid function $\SM$, both visualized in~\cref{fig:act_fcts}. The trainable parameters $\boldsymbol{\mathcal{P}}$ of~\cref{eq:HL1,eq:HL2} are the weights $\bW^{(i)}$, $\bw^{(i)}$ and biases $\bb^{(i)}$. 

In~\cref{sec:param,sec:soft}, we discussed different types of mathematical restrictions required for the construction of the softening function. These restrictions can be incorporated into~\cref{eq:HL1,eq:HL2} as follows:
\begin{itemize}[label=\textbf{--}]
    \setlength{\itemsep}{0pt}
\item \textbf{Zero normalization:} $\cF\rvert_{\bx^{(0)}=\bnull}=0 $ is fulfilled if the $\TANH$ activation function is employed and all biases take zero-values. This is because $\TANH(0)=0$. Note that the zero normalization condition cannot be fulfilled if the $\SM$ activation function is employed, since $\SM(x)>0$ for $x\in\mathbb{R}$.
\item \textbf{Boundedness between zero and one:} $\cF\in[0,1]$ is always fulfilled if the $\SM$ activation function is employed. If the $\TANH$ activation is employed, it is fulfilled if all weights, biases, and the inputs $x^{(0)}_i$ are non-negative. This immediately follows from the values that both functions take, cf.~\cref{fig:act_fcts}. 
\item \textbf{Monotonicity:} $\partial_{x_i^{(0)}}\cF\geq 0$ can be fulfilled for both the $\TANH$ and $\SM$ activation function if all weights are non-negative, cf.~\textcite[Sec.~3.1]{Klein2026a} for a proof thereof. Based on such a monotonically \emph{increasing} function $\cF$, a monotonically \emph{decreasing} function can be constructed as $1-\cF$.
\end{itemize}
For the construction of neural softening functions, we employ different combinations of these constraints.


\subsection{Neural softening functions}\label{sec:neural_softening_models}

Recall that the overall softening function is given by $\eta=1-\eta_1(\psiN,\,\bs)\eta_2(\Dpsi)$, with the following constraints on $\eta_1$ and $\eta_2$ as introduced in~\cref{sec:soft,sec:param}. A sufficient condition for thermodynamic consistency is that $\eta_1\in[0,1]$ and $\partial_{\Dpsi}\eta_2\geq 0$. A sufficient condition for ellipticity is that $\eta_2\in[0,1]$ and $\partial_{\psiN}\eta_1\leq 0$. We further require $\eta_1,\,\eta_2\in[0,1]$ to satisfy our assumption that $\eta\in[0,1]$, and furthermore assume that $\eta_2(0)=0$ should hold. Based on the NN architectures introduced in~\cref{sec:NN_ansatz}, we now introduce different formulations of softening functions to be investigated in this work. 

\begin{figure}[t!]
    \centering        
    \resizebox{0.85\linewidth}{!}{
    \includegraphics{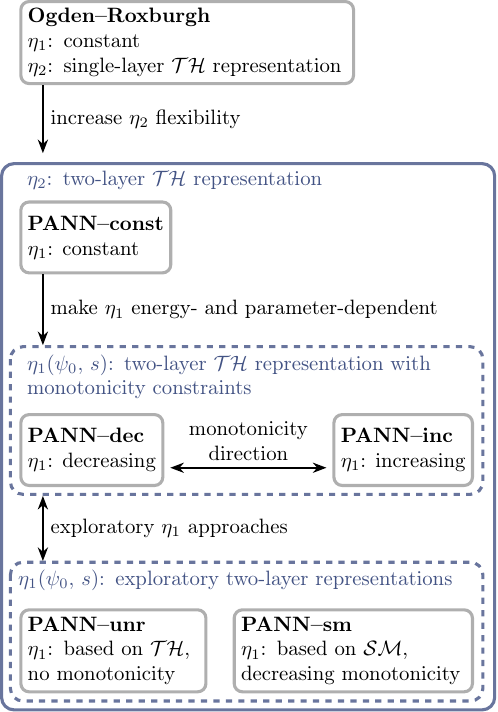}}
 \caption{Conceptual development of the softening model formulations investigated in this work.}
        \label{fig:flowchart}
\end{figure}

\medskip

As a classical reference, we employ the \textbf{Ogden--Roxburgh} model
\begin{equation}\label{eq:OR}
    \eta=1-\frac{1}{r}\TANH\left(\frac{\Dpsi}{m}
    \right) \,,
\end{equation}
with the material parameters $r\geq 1,\,m>0$, cf.~\textcite[Eq.~(3.27)]{ogden_roxburgh}.\footnote{Note that the original model by Ogden~and~Roxburgh employs the error function instead of the hyperbolic tangent, but in the later work by~\textcite[Eq.~(57)]{DORFMANN20041855}, the softening function in~\cref{eq:OR} with the $\TANH$ function was employed.} In the language of machine learning, this model applies a single-layered NN with the $\TANH$ activation function to represent $\eta_2$. Thereby, the parameters $w^{(1)}=\frac{1}{m}\geq 0$ and $b^{(1)}=0$ are constrained to ensure $\eta_2\in[0,1]$, $\partial_{\Dpsi}\eta_2\geq 0$ and $\eta_2(0)=0$. Furthermore, $\eta_1=\frac{1}{r}\in(0,1]$ takes a constant value. 

As discussed in \textcite[Sec.~4.2]{WOLLNER2027106814}, the distinction between classical and NN-based constitutive models is somewhat artificial, as ``[t]hey mainly deviate in the choice of ansatz function made for the constitutive model equations''. The softening functions considered in this work provide an illustrative example of this observation: the classical Ogden--Roxburgh model can itself be interpreted as a single-layered NN ansatz with $\TANH$ activation function. From a technical perspective, NNs are mainly a systematic way of increasing the flexibility of the ansatz function, with which we now proceed. 

\medskip

As a first extension, the \textbf{PANN--const} model retains the constant representation of $\eta_1$, but replaces the single-layered representation of $\eta_2$ by the more flexible two-layered NN with the $\TANH$ activation function, where the parameters are constrained to ensure $\eta_2\in[0,1]$, $\partial_{\Dpsi}\eta_2\geq 0$ and $\eta_2(0)=0$. For all remaining models, we employ this two-layered NN ansatz for $\eta_2$.

\begin{figure}[t!]
    \centering        
    \resizebox{0.87\linewidth}{!}{
    \includegraphics{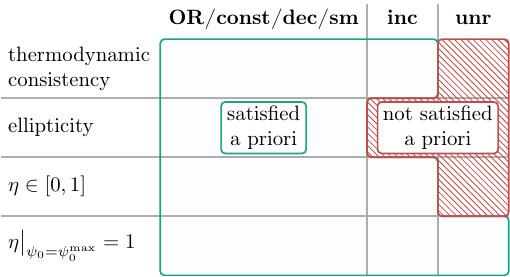}}
 \caption{Fulfillment of constitutive conditions of the softening models investigated in this work, using a polyconvex energy potential. Model labels omit the prefix PANN--; OR denotes the Ogden--Roxburgh model.}
        \label{fig:const_cond}
\end{figure}

\medskip

Next, we introduce an energy- and parameter-dependent representation $\eta_1(\psiN,s)$. For this, the \textbf{PANN--dec} and \textbf{PANN--inc} models employ the same two-layered NN ansatz with the $\TANH$ activation function. In both cases, we employ no zero normalization but boundedness between zero and one. Here, the boundedness between zero and one only holds if $\psiN\geq 0$, which is fulfilled by the employed energy potential, and $\bs\geq 0$, which is satisfied by the employed parameter normalization. The PANN--dec and PANN--inc employ a decreasing and increasing $\eta_1$, respectively.

\medskip

Finally, we introduce two exploratory representations for $\eta_1$, which both employ a two-layered NN ansatz. The \textbf{PANN--unr} model is based on the $\TANH$ function and uses no weight and bias constraints. Consequently, $\eta_1$ does not include the boundedness between zero and one or monotonicity by construction. This model  still satisfies the zero normalization condition in~\cref{eq:eta_norm} through the choice of $\eta_2$. Moreover, the \textbf{PANN--sm} model is based on the $\SM$ function and thus aims to investigate the influence of the employed activation function. Based on $\SM$, this model automatically satisfies the condition of boundedness between zero and one. Moreover, it is restricted to be monotonically decreasing.

\medskip

An overview of the conceptual development of the model formulations and the constitutive conditions they satisfy a priori is provided in~\cref{fig:flowchart,fig:const_cond}. Notably, the Ogden--Roxburgh model and the PANN--const/dec/sm models satisfy all considered constitutive conditions a priori. The PANN--inc model satisfies all conditions besides ellipticity a priori, and the PANN--unr model only satisfies the zero-normalization condition a priori.

\section{Application to experimental data}\label{sec:application}

We now apply the PANN constitutive models discussed in~\cref{sec:PANN} to experimental data of soft polymers. In~\cref{sec:data}, we introduce the experimental datasets. The considered models and calibration strategy are presented in~\cref{sec:calibration}. Subsequently, in~\cref{sec:results_exp}, we investigate the model performance on experimental data.

\subsection{Considered data}\label{sec:data}

We consider experimental data of carbon black-filled ethylene propylene diene rubber materials (EPDM) from~\textcite{Plagge_Klüppel_2017}. EPDM can be synthesized using different fillers and curatives, which results in different properties of the cured material. We consider an EPDM filled with 50~phr N339 carbon black and cured with sulfur and the accelerators CBS and DPG \parencite[Sec.~3.1, Tab.~1]{Plagge_Klüppel_2017}. The amount of curatives is varied, where the CBS content ranges from 0.4 to 2.4~phr, and the sulfur and DPG contents are scaled proportionally \parencite[Sec.~4.1, Fig.~10]{Plagge_Klüppel_2017}. We refer to these compounds as EPDM--$\square$, where $\square\in\{0.4,\,0.9,\,1.4,\,1.9,\,2.4\}$ denotes the CBS content in phr. As the scalar parameter of the PANN models, we employ the linearly scaled curative content $\bs=(\square-0.4)/2.0\in[0,1]$, which parametrizes the considered variation in curative composition.

In~\textcite{Plagge_Klüppel_2017}, uniaxial tension, pure shear, and equibiaxial tension tests are conducted as multihysteresis experiments, i.e., the strain amplitude is successively increased, and each amplitude is applied several times, mostly three to five \parencite[Sec.~4.1, Fig.~9]{Plagge_Klüppel_2017}. As the pseudo-elastic modeling framework that we employ cannot represent changes of the material response in subsequent cycles with the same amplitude, we only consider the first cycle for each new maximum strain amplitude. For each parameter value and load scenario, 1000 datapoints are sampled equidistantly. Apart from this selection, the data is used without further pre-processing. We extract the primary loading path as the envelope of the loading segments up to each new maximum, and the subsequent unloading and reloading paths. Overall, this results in a dataset of the form
\begin{equation}\label{eq:data}
\begin{aligned}
    \cD=\Big\{
    &\big({^1}\lambda,\,{^1}\lambda^{\max},\,{^1}P,\,{^1}s\big),\,\dotsc,
    \\
    &\big({^m}\lambda,\,{^m}\lambda^{\max},\,{^m}P,\,{^m}s\big)
    \Big\}\,,
\end{aligned}
\end{equation}
consisting of $m$ data points. Here, $\lambda$ and $P$ denote the stretch and the first Piola--Kirchhoff stress in tensile direction, and $\lambda^{\max}$ denotes the maximum stretch attained in the preceding deformation history of the respective load path. The data is split into calibration, validation, and test datasets as summarized in~\cref{tab:split}. Following standard machine learning practice \parencite{herrmann2025}, the calibration data is used to fit the model parameters, the validation data is used solely for selection of models after training, and the test data is used only for the final model evaluation. As pure shear lies between uniaxial and equibiaxial tension in the invariant plane, the pure shear test cases correspond to an interpolation in the deformation space \parencite{Klein2026a}.

\begin{table}[t]
\centering
\caption{Split of the experimental data into calibration (C), validation (V), and test (T) datasets.}
\label{tab:split}
\begin{tabular}{lccccc}
\toprule
& \multicolumn{5}{c}{EPDM} \\
\cmidrule(lr){2-6}
& 0.4 & 0.9 & 1.4 & 1.9 & 2.4 \\
\midrule
uniaxial tension    & C & T & C & V & C\\
pure shear          & T & T & T & T & T\\
equibiaxial tension & C & V & C & T & C\\
\bottomrule
\end{tabular}
\end{table}

\begin{remark}\label{rem:data_discussion}
While the pseudo-elastic constitutive models employed in this work allow for a fairly accurate representation of the considered experimental data, our modeling framework is still a simplification. As discussed in~\textcite{Plagge_Klüppel_2017}, the considered EPDM exhibits further effects such as rate- and temperature-dependency, and softening during repeated cycles with a fixed $\psiM$. Such effects cannot be represented within the pseudo-elastic modeling framework we consider, which clearly constitutes a limitation of our approach. Nevertheless, by reducing the model complexity, we aim to gain insights into appropriate formulations of neural softening functions. Ultimately, such investigations on simplified constitutive models provide the foundation for developing more sophisticated modeling approaches including further relevant effects.
\end{remark}

\subsection{Considered models and calibration strategy}\label{sec:calibration}

We now provide information on the considered pseudo-elastic constitutive models and the model calibration strategy. 

\paragraph{Considered models.} 
For all pseudo-elastic constitutive models, the virgin energy potential $\psiN$ is represented by the physics-augmented neural network introduced in~\cref{eq:NN_energy}. Following~\textcite[Sec.~4.2]{Klein2026a}, we employ $n=8$ nodes in each hidden layer, resulting in 112 trainable parameters. This energy potential is combined with the different softening functions $\eta_1,\,\eta_2$ introduced in~\cref{sec:neural_softening_models}. Based on preliminary hyperparameter studies, all neural softening models likewise employ $n=8$ nodes in each hidden layer, cf.~\cref{sec:NN_ansatz}. For the softening function, this results in 2 trainable parameters for the Ogden--Roxburgh model, 17 for PANN--const, and 49 for PANN--inc, PANN--dec, PANN--unr, and PANN--sm. This pragmatic hyperparameter choice for the neural softening function is supported by the observation that once a NN architecture is sufficiently flexible for the considered material behavior, there is no benefit in further increasing the parameter count~\parencite[Sec.~4.2]{Klein2026a}.

\paragraph{Stress response.} We need to evaluate the stress response of the pseudo-elastic modeling framework introduced in~\cref{sec:basics} for the load scenarios employed in the experimental investigations, cf.~\cref{sec:data}. Assuming incompressibility, the deformation gradient for these load scenarios is given by
\begin{equation}\label{eq:def_grads}
\begin{aligned}
\bF_{\text{uniaxial}}&=\operatorname{diag}(\lambda,\,\lambda^{-1/2},\,\lambda^{-1/2})\,,
\\
 \bF_{\text{equibiaxial}}&=\operatorname{diag}(\lambda,\,\lambda,\,\lambda^{-2})\,,
 \\
 \bF_{\text{pure shear}}&=\operatorname{diag}(\lambda,\,1,\,\lambda^{-1})\,.
\end{aligned}
\end{equation}
The corresponding stress tensor $\bP$ is likewise diagonal, where the 11-component corresponds to the stress in tensile direction, and the 33-component -- and, for uniaxial tension, also the 22-component -- vanishes. These vanishing stress components enable to explicitly calculate the Lagrange multiplier $p$ in~\cref{eq:total_stress}, and, consequently, to calculate the stress in tensile direction as
\begin{equation}\label{eq:model_stress_prediction}
    P^{\text{model}}=\eta\,\big(\partial_{\lambda_1}\psiN-\lambda_3/\lambda_1\,\partial_{\lambda_3}\psiN\big)
\end{equation}
where the $\lambda_i$ are chosen according to~\cref{eq:def_grads}, see e.g.~\textcite{Steinmann_Hossain_Possart_2012}.

\medskip

\paragraph{Two-stage calibration.}
As commonly done in pseudo-elasticity, see e.g.~\textcite{ogden_roxburgh}, all models are calibrated in two stages, where the virgin energy potential $\psiN$ is calibrated first and subsequently kept fixed while calibrating the softening function. In both cases, we minimize the mean squared error (MSE) loss function
\begin{equation}\label{eq:loss}
\begin{aligned}
        \mathfrak{L}_{\star}(\boldsymbol{\mathcal{P}})=\frac{1}{m_{\star}}\sum_{i=1}^{m_{\star}}\left(^iP-{^iP}^{\text{model}}({^i\lambda},\,{^i}\lambda^{\max},\,{^i}\bs;\,\boldsymbol{\mathcal{P}})\right)^2\,,
\end{aligned}
\end{equation} 
on a dataset $\cD_{\star}$ of size $m_{\star}$. Furthermore, $P^\text{model}$ is calculated according to~\cref{eq:model_stress_prediction} and $\boldsymbol{\mathcal{P}}$ denotes the trainable parameters of the NN. We employ no loss weighting. 

In the first stage, the dataset used to evaluate the loss function in~\cref{eq:loss} consists only of the primary loading paths of the overall calibration dataset, i.e., $\lambda=\lambda^{\max}$. Then, the normalization condition in~\cref{eq:eta_norm} yields $\eta=1$, such that the stress in~\cref{eq:model_stress_prediction} depends only on the energy potential. The energy potential is calibrated five times with randomly initialized parameters, and the model instance with the lowest validation loss is used in all subsequent investigations. Consequently, the energy potential is identical for all softening models, and differences between the models stem from the softening function alone.

In the second calibration stage, the parameters of the energy potential are fixed, and the overall calibration dataset is used to evaluate the loss function in~\cref{eq:loss}, now including also unloading and reloading data points. To assess the robustness with respect to parameter initialization, each softening model is calibrated 100 times from randomly initialized parameters. Unless stated otherwise, the instance with the lowest validation loss is used for evaluation.  

\begin{remark}
Alternatively, in the second calibration stage, the softening function could be determined directly from the calibrated energy potential and the experimental data, enabling a loss formulation in terms of the softening function itself. However, this approach requires division by the stress predicted by the energy potential, leading to numerical instabilities at small stress values. It was therefore not considered further.
\end{remark}

\paragraph{Implementation.}
All models are implemented in \textsc{KLAX} 0.2 (see~\url{https://drenderer.github.io/klax}) and calibrated using the SLSQP optimizer with a maximum of 2000 iterations. A single calibration of a neural energy potential or a neural softening function takes a few seconds on a MacBook Air M4 with 16 GB of RAM.

\subsection{Model performance on experimental data}\label{sec:results_exp}

\begin{figure}[t]
\centering
\includegraphics[width=0.99\columnwidth]{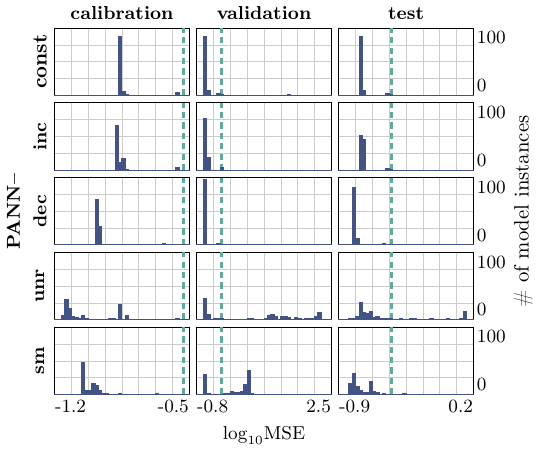}
 \caption{Distribution of the MSE values for 100 repeated calibrations of each softening model. The energy potential is identical in all cases, and only the softening function is varied. The Ogden--Roxburgh model showed practically no variation between different calibration runs; for comparison, it is indicated by a dashed vertical line.}
\label{fig:histogram}
\end{figure}

\begin{figure*}[t]
\centering
    \includegraphics[width=0.95\textwidth]{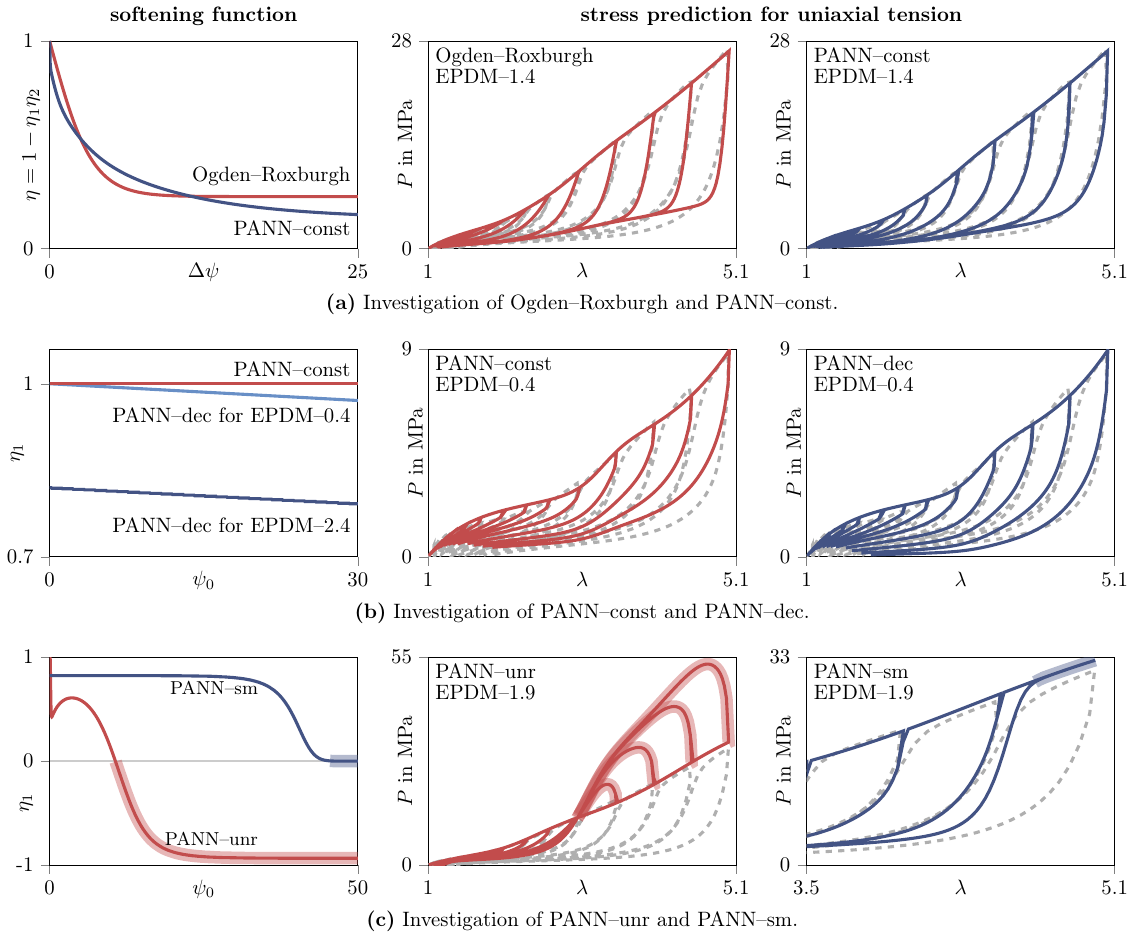}
 \caption{Comparison of different softening function representations and their implications on stress prediction for uniaxial tension. Dashed lines denote experimental data and continuous lines model predictions.
 }
\label{fig:comparison_figure}
\end{figure*}

In the first calibration stage, the neural energy potential is calibrated to the primary loading paths of the calibration dataset, for which the employed model shows a very good performance. As a nearly identical PANN model was already applied to the same data in~\textcite[Sec.~4.3]{Klein2026a}, we refer to this work for a detailed discussion. In the second calibration stage, each softening model is calibrated 100 times with randomly initialized parameters, while the energy potential is kept fixed. The distributions of the resulting MSE values on the calibration, validation, and test datasets are visualized in~\cref{fig:histogram}. Overall, the PANN--const, PANN--inc, and PANN--dec models show narrow MSE distributions and clearly outperform the Ogden--Roxburgh model on all datasets. In contrast, the exploratory PANN--unr and PANN--sm models show a wide spread of MSE values, in particular on the validation and test datasets. In the following, we discuss the different models in the order of their conceptual development, cf.~\cref{fig:flowchart}.

\paragraph{Increase of $\eta_2$-flexibility.} 
To investigate the influence of increasing the flexibility of $\eta_2$, we compare the Ogden--Roxburgh model and the PANN--const model, with the MSE values visualized in~\cref{fig:histogram} and exemplary softening function and stress predictions for EPDM--1.4 visualized in~\cref{fig:comparison_figure}a. 
The Ogden--Roxburgh model shows the largest MSE values on all datasets. Still, its stress predictions are reasonably accurate, although they show visible deviations from the experimental data, in particular for the softening unloading and reloading paths. By replacing the single-layered representation of $\eta_2$ with a two-layered NN, the PANN--const model considerably reduces the MSE values, which is also reflected in an improved stress prediction quality. The reason for this can be seen in the calibrated softening functions. While the softening function of the Ogden--Roxburgh model is restricted to a single $\TANH$ function, which saturates rapidly, the PANN--const model can represent more general relationships, such as a steep initial decrease for small values of $\Dpsi$ combined with a continued slow decrease for larger values. Notably, both models perform remarkably well, considering that their softening functions do not depend on the parameter $\bs$. However, this does not imply that the softening behavior of the material is independent of $\bs$. As the softening function is evaluated on $\Dpsi$, which is measured in terms of the parametrized energy potential $\psiN(\,\cdot\,,\bs)$, the softening in terms of the stress response still depends on the parameter. Rather, the results suggest that, when expressed in terms of the energy measure $\Dpsi$, the softening behavior of the considered EPDM depends only weakly on the curative content.

\begin{figure}[t]
\centering
\includegraphics[width=0.97\columnwidth]{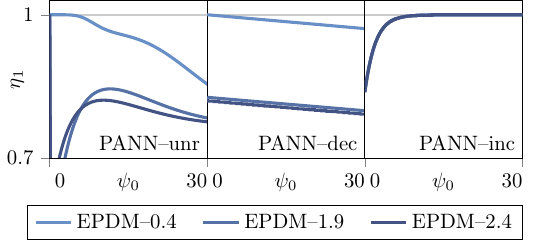}
 \caption{Investigation of different $\eta_1$ representations evaluated on multiple parameter values.}
\label{fig:eta1_comparison}
\end{figure}

\begin{figure*}[t!]
\centering
    \includegraphics[width=0.9\textwidth]{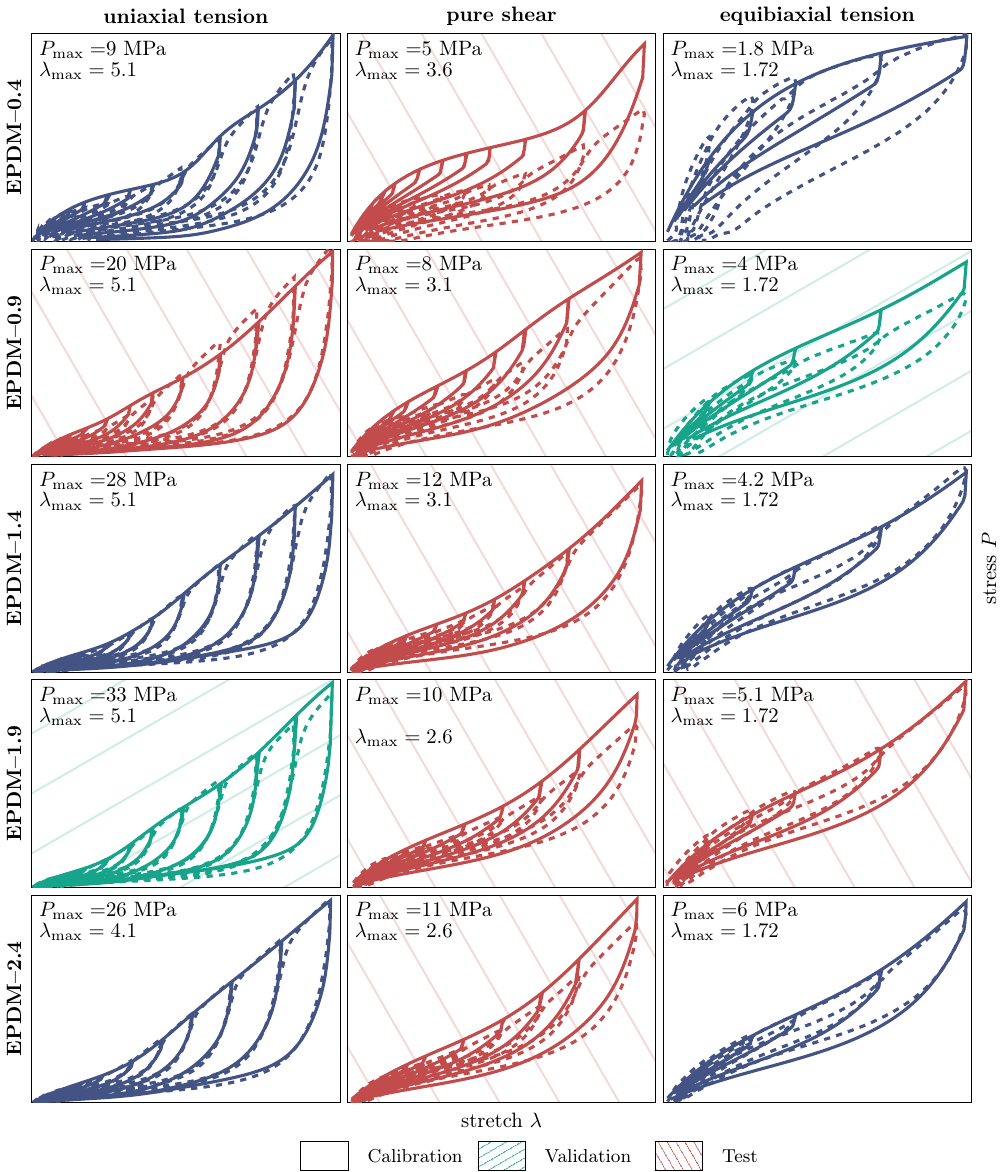}
 \caption{Stress predictions of the PANN--dec model for all load cases and parameter values. Dashed lines denote experimental data and continuous lines model predictions. The background patterns and colors denote calibration, validation, and test data, cf.~\cref{tab:split}.}
\label{fig:PANN_dec_overall}
\end{figure*}

\paragraph{Energy- and parameter-dependent $\eta_1$.} 
Next, we consider the PANN--inc and PANN--dec models, in which $\eta_1$ depends on both the energy potential and the parameter. The MSE values are visualized in~\cref{fig:histogram} and exemplary softening function and stress predictions of the PANN--dec model for EPDM--0.4 are visualized in~\cref{fig:comparison_figure}b. 
While the PANN--inc model performs approximately as well as the PANN--const model, the PANN--dec model slightly improves the performance on the calibration dataset, which is reflected in an improved stress prediction for some parameter values. The calibrated $\eta_1$ functions show that both models exploit their additional flexibility in different ways, cf.~\cref{fig:eta1_comparison}. The PANN--dec model mainly exploits its parametric dependency, with $\eta_1\approx 0.97$ for EPDM--0.4 and $\eta_1\approx 0.82$ for EPDM--1.9 and EPDM--2.4, and only a weak decrease with $\psiN$. In contrast, the PANN--inc model learns a function which is practically independent of $\bs$. As the employed NN ansatz functions couple the monotonicity in $\bs$ to the monotonicity in $\psiN$, cf.~\cref{sec:param}, the PANN--inc model cannot represent a softening which decreases with the curative content. Instead, it exploits its dependency on the energy potential: $\eta_1$ increases from approximately $0.85$ at $\psiN=0$ MPa to $\eta_1\approx 1$ for $\psiN\geq 5$ MPa and remains constant for larger energies. Thus, the PANN--inc model reduces the softening for small energy values, while it behaves similarly to the PANN--const model for larger energies. 

At first glance, it might appear surprising that models with opposite monotonicity in $\psiN$ perform similarly. The PANN--unr model, which is not restricted in its monotonicity, provides an explanation, cf.~\cref{fig:eta1_comparison}. For EPDM--1.9 and EPDM--2.4, its $\eta_1$ increases with $\psiN$ for small energy values, similar to the PANN--inc model, and decreases for larger energy values, similar to the PANN--dec model. Furthermore, it shows a pronounced parametric dependency, again similar to the PANN--dec model. This suggests that the softening behavior in the data comprises different features in different energy ranges, which cannot be represented by a single monotonicity direction. Consequently, the PANN--inc and the PANN--dec model each capture only part of these features, which explains their similar overall performance. Although the PANN--unr model can represent both features simultaneously, it does so without any guarantee of a physically reasonable behavior, as we discuss next.

\paragraph{Exploratory models.} Finally, we discuss the PANN--unr and PANN--sm models, for which the MSE values are visualized in~\cref{fig:histogram} and exemplary stress predictions for EPDM--1.9 in~\cref{fig:comparison_figure}c. Note that in~\cref{fig:comparison_figure}c, we do not show the model instances with lowest validation MSE, but rather models which exemplify typical failure modes. The PANN--unr model shows a wide spread of MSE values. While some instances achieve lower calibration MSE values than the constrained models, only about a third of the instances yield physically reasonable predictions, which is indicated by low validation MSE values. The remaining instances show unphysical behavior when evaluated outside the calibration data, which is caused by a lack of structure for extrapolation. In particular, these instances predict negative values of $\eta_1$ for energy values beyond the calibration regime, which leads to an increase of the stress upon unloading. Thus, an unrestricted NN architecture can, but does not reliably, yield a physically reasonable softening function.

The PANN--sm model fulfills the same constitutive conditions as the PANN--dec model and mainly differs in the choice of activation function. Nevertheless, many of its instances fail for EPDM--1.9, which belongs to the validation dataset, and show a rapid collapse of the softening function outside the calibration regime. This can be explained by the regime in which the activation functions are evaluated, see~\cref{fig:act_fcts} for a visualization of the activation functions. In the $\TANH$-based NNs, the non-negative weights, biases, and inputs restrict each node to the concave, saturating regime of the activation function, leading to a conservative extrapolation behavior. In contrast, in the $\SM$-based NNs, the biases are unconstrained, as boundedness is fulfilled for arbitrary biases. Thus, the inflection point of the activation function can be shifted to just outside the calibration regime, leading to the sharp decrease of $\eta_1$ observed in~\cref{fig:comparison_figure}c.

\begin{remark}\label{rem:ansatz_shape}
At first glance, the choice between the $\TANH$ and $\SM$ activation function might appear to be of minor relevance, as both NN architectures fulfill the same constitutive conditions introduced in~\cref{sec:basics}. However, the results above demonstrate that fulfilling the same constraints is not sufficient. In particular, the specific shape of the ansatz function and the regime in which it is evaluated have a pronounced influence on the model behavior, in our case particularly outside the calibration regime. In classical constitutive modeling, it is standard to choose ansatz functions based on their specific shape. Clearly, PANN models do not remove the need for well-chosen ansatz functions but merely shift this choice to the level of the NN architecture.
\end{remark}

\paragraph{Overall performance of PANN--dec.} Finally, \cref{fig:PANN_dec_overall} shows the stress predictions of the PANN--dec model for all load cases and parameter values. For uniaxial tension, the model accurately represents both the primary loading paths and the unloading and reloading paths, including the validation and test cases EPDM--1.9 and EPDM--0.9, respectively, which correspond to an interpolation in the parameter space. Likewise, the pure shear predictions, which are test data for all parameter values and correspond to an interpolation in the deformation space, show good agreement with the experimental data. For larger curative contents, the prediction of equibiaxial tension is also excellent. The most pronounced deviations occur for unloading or reloading of EPDM--0.4 under equibiaxial tension. A possible explanation is that, as the softening function depends on the deformation only through the energy potential, equal values of $\psiN$ and $\psiM$ result in equal softening regardless of the load case, which may limit the representation of load case-dependent softening. Overall, the PANN--dec model, which satisfies all considered constitutive conditions a priori, including the sufficient condition for ellipticity, provides accurate predictions over a wide range of deformation states and parameter values.


\section{Conclusion}\label{sec:conc}

In this work, we investigated parametrized physics-augmented neural network constitutive models for the Mullins effect within the pseudo-elastic modeling framework, focusing on structural constraints on neural softening functions. We introduced a multiplicative decomposition of the softening function that enables the incorporation of parametric dependencies while satisfying common constitutive conditions by construction. Furthermore, we derived a condition for ellipticity of pseudo-elastic models, which can be incorporated by sufficient (but not necessary) conditions such as polyconvexity of the energy potential and a monotonicity constraint on the neural softening function. We systematically investigated neural softening functions with varying flexibility and constraints using multiaxial experimental data of carbon black-filled EPDM rubber with varying curative content from \textcite{Plagge_Klüppel_2017}. The main findings are as follows:
\begin{itemize}[label=\textbf{--}]
    \setlength{\itemsep}{0pt}
    \item Neural softening functions are a natural extension of the classical pseudo-elastic modeling framework, improving the representation of materials that experience Mullins softening while retaining common constitutive conditions. Moreover, neural softening functions offer the potential for a unified modeling framework that reduces the need to select an appropriate classical softening function for a given material \parencite{ricker2021}. Similar advantages of neural constitutive models over classical approaches have been demonstrated for hyperelasticity in~\parencite{Klein2026a}.
    \item Structural constraints on the softening function are essential for reliable predictions. While unconstrained models can achieve low calibration errors, they frequently exhibit physically unreasonable predictions. In contrast, constrained formulations exclude such behavior by construction and exhibit narrow error distributions across repeated calibrations. Similar observations have been reported for hyperelastic PANN models in~\parencite{linden2023,Klein2026a}.
    \item Even if two constitutive models are similar in that they satisfy the same constitutive conditions, they can have a vastly different performance. We demonstrate this at the example of softening functions based on the hyperbolic tangent function and based on the sigmoid function, respectively. Similar observations have been reported for hyperelastic PANN models in~\parencite{WOLLNER2027106814,KLEIN2027106815,Geuken2026}.
\end{itemize}
Future work should address additional aspects of the Mullins effect, including temperature dependence \parencite{Plagge_Klüppel_2017}, recovery over time \parencite{LIAO2023104338}, softening-induced anisotropy \parencite{DENORA2024108399}, and coupling with other inelastic phenomena such as viscoelasticity \parencite{ALKHOURY2024105650}. Furthermore, the proposed parametrized PANN formulations should be investigated using additional experimental datasets to gain further insights on how neural softening functions should be formulated for different material behaviors.


\vspace*{3ex}

\noindent
\textbf{CRediT authorship contribution statement.} 
\textbf{D.K.\ Klein:} Conceptualization, Methodology, Software, Formal analysis, Investigation, Writing -- original draft, Writing -- review and editing, Visualization. 
\textbf{J.\ Fey:} Conceptualization, Methodology, Software, Investigation, Writing -- review and editing. 
\textbf{M.\ Kannapinn:} Software, Investigation, Writing -- review and editing. 
\textbf{O. Weeger:} Conceptualization, Funding acquisition, Writing -- review and editing. 
\vspace*{1ex}

\noindent
\textbf{Conflict of interest.} The authors declare that they have no conflict of interest.
\vspace*{1ex}

\noindent
\textbf{Acknowledgment.} 
D.K.\ Klein and O.\ Weeger acknowledge the financial support provided by the Deutsche Forschungsgemeinschaft (DFG, German Research Foundation, project number 492770117) and the Graduate School of Computational Engineering at TU Darmstadt.
\vspace*{1ex}

\noindent
\textbf{Data availability.}
The authors have no data to share.
\vspace*{1ex}

\noindent
\textbf{AI declaration.} During the preparation of this work, the authors used Claude (Anthropic) and ChatGPT (OpenAI) to support language editing, improve clarity and readability, and assist with code development and debugging. All scientific content, methodological choices, numerical results, and conclusions were developed, verified, and approved by the authors. The authors take full responsibility for the content of the manuscript.


\appendix
\numberwithin{equation}{section} 

\section{Ellipticity in pseudo-elastic Mullins modeling}\label{sec:ellip}

We now derive the ellipticity condition for the incompressible pseudo-elastic Mullins softening framework introduced in~\cref{sec:basics}. For convenience of notation, we omit the dependency of the material model on the parameter $\bs$ in the following, and note that ellipticity should hold for arbitrary values of $\bs$. 

Recall that the first Piola-Kirchhoff stress of this framework is given by
\begin{equation}\label{eq:total_stress_recall}
\bP=\eta \,\derivF\psiN -p \bH\,,
\end{equation}
where $\psiN$ is the unsoftened energy potential and $\eta$ the softening function. In the following, we start by assuming a general softening function that depends on $\psiN$ and $\psiM$, cf.~\cref{eq:psiM}, and subsequently consider the case of a multiplicative decomposition of the softening function as
\begin{equation}\label{eq:psi_12_recall}
    \eta=1-\eta_1(\psiN)\,\eta_2(\Dpsi)
\end{equation}
where $\Dpsi=\psiM-\psiN$, cf.~\cref{eq:Dpsi}.

We begin by differentiating between different loading paths:
\begin{definition}
    We refer to a loading path as \emph{unloading} or \emph{reloading} if the material is in its elastic regime. This means that either $\psiN<\psiM$, or that $\psiN=\psiM$ with the incremental change of $\psiM$ along the loading path satisfying $\mathrm{D}\psiM=0$. We refer to a loading path as \emph{primary} if softening occurs, i.e., $\psiN=\psiM$ with the incremental change of $\psiM$ along the loading path satisfying $\mathrm{D}\psiM=\mathrm{D}\psiN>0$. 
\end{definition}
Due to the normalization condition~\cref{eq:eta_norm}, it holds that $\eta\leq 1$ in the unloading or reloading regime and $\eta=1$ on primary load paths. 

This distinction is required to calculate the consistent tangent modulus, see e.g.\ \textcite[Rem.~4.1]{DESOUZANETO19941533}:
\begin{theorem}\label{th:consistent_tang}
For the pseudo-elastic modeling framework introduced in~\cref{sec:basics}, the consistent tangent modulus
\begin{equation}
    \bbA=\derivF\big(\eta\,\derivF\psiN\big)
\end{equation}
is given by
\begin{equation}
\bbA=
\eta\,\derivFF\psiN+\big(\partial_{\psiN}\eta\big)\,\derivF\psiN\otimes\derivF\psiN\,,
\end{equation}
for unloading or reloading paths, and by
\begin{equation}
\bbA=\derivFF\psiN \,,
\end{equation}
for primary load paths.
\end{theorem}
\begin{proof}
For unloading or reloading, $\psiM$ is constant, and $\bbA$ follows immediately from the product and chain rules. For primary load paths, $\psiM$ is not constant, and we have
\begin{equation}\label{eq:tangent_damage}
\bbA=\big(\partial_{\psiN}\eta+\partial_{\psiM}\eta\big)\,\derivF\psiN\otimes\derivF\psiN+\eta\,\big(\derivFF\psiN\big)\,.
\end{equation}
As we assume that $\eta=1=\text{const.}$ on primary load paths, cf.~\cref{eq:eta_norm}, the increment of $\eta$ satisfies
\begin{equation}
\begin{aligned}
\mathrm{D}\eta&=\partial_{\psiN}\eta\,D\psiN+\partial_{\psiM}\eta\,D\psiM
\\
&=\big(\partial_{\psiN}\eta+\partial_{\psiM}\eta\big)\,D\psiN=0\,.    
\end{aligned}
\end{equation}
This implies that $\partial_{\psiN}\eta+\partial_{\psiM}\eta=0$ for primary load paths. Together with~\cref{eq:tangent_damage}, this provides $\bbA$ for primary load paths.
\end{proof}

Based on this, we can derive the ellipticity condition for pseudo-elasticity:

\begin{theorem}\label{th:ellip}
The ellipticity condition for the considered pseudo-elastic modeling framework is given by
\begin{equation}\label{eq:ellip}
\begin{aligned}
\ab:\bbA:\ab    \geq 0 \\
\forall\ba,\bb\in\bbR^3\,,\bF+\ba\otimes\bb\in\operatorname{SL}(3)\,,    
\end{aligned}
\end{equation}
with the consistent tangent modulus $\bbA$ for unloading / reloading paths and primary load paths as derived in~Theorem~\ref{th:consistent_tang}. Physically, \cref{eq:ellip} corresponds to the existence of real-valued wave speeds.
\end{theorem}
\begin{proof}
Consider the homogeneous equilibrium state $\bxeq=\bxn+\bueq$, where $\bFeq$ and $\peq$ designate the associated equilibrium deformation gradient and pressure fields with $\det\bFeq=1$. The linearization of the stress at the equilibrium state is given by
\begin{equation}
\begin{aligned}
        \operatorname{Lin}\bP=&{\bPeq}
    +\bbAeq:\Delta\bF
    \\
&-\Delta p\bHeq - \peq\, \Delta \bF \Cross \bFeq\,,
\end{aligned}
\end{equation}
with the stress at the equilibrium state
\begin{equation}
    \bPeq=\eta \,\derivF\psiN\big\rvert_{\bF=\bFeq} -\peq \bHeq\,,
\end{equation}
and the consistent tangent modulus 
\begin{equation}
    \bbAeq=\bbA\rvert_{\bF=\bFeq}\,,
\end{equation}
where $\bbA$ depends on whether unloading / reloading or the softening regime is considered, cf.~Theorem~\ref{th:consistent_tang}. For the remaining calculations, the procedure known for the derivation of ellipticity in incompressible hyperelasticity can be employed, see e.g.\ \textcite{Zee1983}, the only difference being that the expression for the tangent operator $\bbA$ differs between hyperelasticity and pseudo-elastic Mullins modeling. The remaining steps can be outlined as follows, see e.g.\ \textcite[App.~A]{Klein2026a}:
\begin{enumerate}
    \item In the absence of body forces, the linearized balance of linear momentum is given by 
    \begin{equation}
        \operatorname{Div}[\bbAeq:\Delta\bF-\Delta p\,\bHeq]=\rho_0\Delta\ddot{\bu}\,.
    \end{equation}
    \item Next, we seek solutions in the form of traveling waves
    \begin{equation}
        \begin{aligned}
            \Delta\bu=\Bar{\bu}\,f(\bx_0\cdot\bn_0-ct)\,,
        \\
        \Delta p=\Bar{p}\,f(\bx_0\cdot\bn_0-ct)\,,
        \end{aligned}
    \end{equation}
    where the unit vector $\bn_0\in\bbR^3$ denotes the direction of wave propagation, $\Bar{\bu}\in\bbR^3\backslash\{\bnull\}$ and $\Bar{p}\in\bbR$ denote the eigenmodes of the deformation, $c\in\bbR$ is the wave speed, and $f$ a sufficiently smooth function.
    \item Then, the wave ansatz is inserted into the linearized balance of linear momentum. After some algebraic manipulations, this provides the ellipticity condition in~\cref{eq:ellip}, cf.~\parencite[App.~A]{Klein2026a}.
\end{enumerate}
\end{proof}
\begin{theorem}\label{th:ellip_specific}
Sufficient (but not necessary) conditions for the considered pseudo-elastic modeling framework to be elliptic are that
\begin{enumerate}[label=(\roman*)]
    \setlength{\itemsep}{0pt}
\item the energy potential $\psiN$ is elliptic, e.g., by using a polyconvex $\psiN$, 
\item the softening function satisfies $\eta\geq 0$ and the monotonicity constraint $\partial_{\psiN}\eta \geq 0$.
\end{enumerate}
\end{theorem}
\begin{proof}
We start by investigating ellipticity for unloading or reloading paths, for which we have to show that
\begin{equation}\label{eq:ellip_special_1}
\begin{aligned}
\ab:\Big[&\eta\,\big(\derivFF\psiN\big)
\\
&+\big(\partial_{\psiN}\eta\big)\derivF\psiN\otimes\derivF\psiN\Big]:\ab
\\
&\geq 0\quad \forall\ba,\bb\in\bbR^3\,,\bF+\ba\otimes\bb\in\operatorname{SL}(3)\,.
\end{aligned}
\end{equation}
Under the assumption (ii), sufficient conditions for \cref{eq:ellip_special_1} are given by
\begin{equation}\label{eq:ellip_elastic}
\begin{aligned}
\ab:\derivFF\psiN:\ab&\geq 0 
\\
\text{and}\qquad \ab:\Big[\derivF\psiN\otimes\derivF\psiN\Big]:\ab&\geq 0 
\\
\forall\ba,\bb\in\bbR^3\,,\bF+\ba\otimes\bb\in\operatorname{SL}(3)\,.
\end{aligned}
\end{equation}
The condition in~\cref{eq:ellip_elastic}$_1$ is fulfilled by assumption (i). The condition in~\cref{eq:ellip_elastic}$_2$ can be expressed as $\left(\derivF\psiN:\ab\right)^2\geq 0$ and is also fulfilled. Thus, the overall material model is elliptic for unloading or reloading paths. 

On primary load paths, the ellipticity condition is given by
\begin{equation}\label{eq:ellip_damage}
\begin{aligned}
\ab:\big(\derivFF\psiN\big):\ab\geq 0
\\
\forall\ba,\bb\in\bbR^3\,,\bF+\ba\otimes\bb\in\operatorname{SL}(3)\,,
\end{aligned}
\end{equation}
which is fulfilled under assumption (i). Thus, the material model is elliptic for both unloading / reloading and primary load paths.
\end{proof}
Theorem~\ref{th:ellip_specific} can immediately be applied for softening functions with a multiplicative decomposition.
\begin{corollary}\label{corr:ellip_multi}
Let us consider a pseudo-elastic constitutive model with a multiplicatively decomposed softening function of the form $\eta=1-\eta_1(\psiN)\,\eta_2(\Dpsi)$. Sufficient (but not necessary) conditions for such a model to be elliptic are that 
\begin{itemize}[label=\textbf{--}]
    \setlength{\itemsep}{0pt}
\item $\psiN$ is a polyconvex function.
\item $\eta_1$ satisfies $\partial_{\psiN}\eta_1\leq 0$ and $\eta_1\in[0,1]$.
\item $\eta_2$ satisfies $\partial_{\Dpsi}\eta_2\geq 0$ and $\eta_2\in[0,1]$.
\end{itemize}
\end{corollary}
\begin{proof}
This follows immediately from Theorem~\ref{th:ellip_specific}. To show that assumption (ii) from Theorem~\ref{th:ellip_specific} holds, we observe that
\begin{equation}
  \eta_1\eta_2\leq 1\quad \Rightarrow\quad  \eta=1-\eta_1\eta_2\geq 0\,,
\end{equation}
and
\begin{equation}
    \partial_{\psiN}\eta=-\big(\partial_{\psiN}\eta_1\big)\,\eta_2-\eta_1\,\big(\deriv_{\psiN}\eta_2\big)\geq 0\,,
\end{equation}
where the latter holds because $\eta_1,\eta_2\geq 0$, $\partial_{\psiN}\eta_1\leq 0$, and $\deriv_{\psiN}\eta_2=\big(\partial_{\Dpsi}\eta_2\big)\,\big(\partial_{\psiN}\Dpsi\big)=
-\partial_{\Dpsi}\eta_2\leq 0$.
\end{proof}
\begin{remark}
    As a consequence of Corollary~\ref{corr:ellip_multi}, several of the PANN models introduced throughout this manuscript are elliptic by construction, including the PANN--const, PANN--dec, and PANN--sm models. Likewise, many classical approaches such as the model proposed by \textcite{ogden_roxburgh} are elliptic, given that they employ a polyconvex energy potential.
\end{remark}

\renewcommand*{\bibfont}{\footnotesize}
\printbibliography

\end{document}